\documentclass[12pt]{article}

\usepackage{amsmath}
\usepackage{amssymb}
\usepackage{amsfonts}
\usepackage{latexsym}
\usepackage{color}
\usepackage{graphicx}

\catcode `\@=11 \@addtoreset{equation}{section}

\catcode `\@=12

\newcommand{\be}{\begin{equation}}
	\newcommand{\en}{\end{equation}}
\newcommand{\bea}{\begin{eqnarray}}
	\newcommand{\ena}{\end{eqnarray}}
\newcommand{\beano}{\begin{eqnarray*}}
	\newcommand{\enano}{\end{eqnarray*}}
\newcommand{\bee}{\begin{enumerate}}
	\newcommand{\ene}{\end{enumerate}}

\newcommand{\B}{{\mathfrak B}}

\newcommand{\mc}{\mathcal}

\newcommand{\D}{{\mc D}}

\newcommand{\Sc}{{\cal S}}

\newcommand{\F}{{\cal F}}

\newcommand{\Lc}{{\cal L}}

\newcommand{\1}{1 \!\! 1}

\newcommand{\ltwo}{{\cal L}^2(\mathbb{R})}

\newcommand{\Hil}{\mc H}

\newtheorem{thm}{Theorem}

\newtheorem{lemma}[thm]{Lemma}
\newtheorem{prop}[thm]{Proposition}

\newenvironment{proof}{\noindent {\bf Proof --}}{\hfill$\square$ \vspace{3mm}\endtrivlist}

\catcode `\@=11 \@addtoreset{equation}{section}
\catcode `\@=12

\begin{document}
	
	\thispagestyle{empty}

	\vspace*{2cm}
	
	\begin{center}
		{\Large \bf Uncertainty relation for non-Hermitian operators}   \vspace{2cm}\\
		
		{\large F. Bagarello}\\
		Dipartimento di Ingegneria,\\[0pt]
		Universit\`{a} di Palermo, I - 90128 Palermo,\\
		and I.N.F.N., Sezione di Catania\\
		E-mail: fabio.bagarello@unipa.it\\

		\vspace{2mm}

	\end{center}
	
	\vspace*{2cm}
	
	\begin{abstract}
		\noindent 	In this paper we discuss some aspects of the Heisenberg uncertainty relation,  mostly from the point of view of non self-adjoint operators. Some equivalence results, and some refinements of the inequality, are deduced, and some relevant examples are discussed.
		
		We also begin a sort of {\em dynamical analysis} of the relation, in connection with what has been recently called $\gamma$-{dynamics} and $\gamma$-symmetries, and we discuss in some details the role of different scalar products in our analysis.
		
		The case of self-adjoint operators is recovered as a special case of our general settings. 		
	\end{abstract}
	
	\vspace{2cm}
	
	
	\vfill


	\newpage
	
	\section{Introduction}
	
	The Heisenberg uncertainty relation (HUR), \cite{heis,robe,schr}, has always been one of the most interesting results connected with Quantum Mechanics, especially when compared with Classical Mechanics. The reason is simple: more than its mathematical derivation, which is not hard, its physical  consequences were rather new, and unexpected. In fact, the HUR implies that it is just  impossible to perform an exact measure of two (non-commuting) observables, independently of the precision of the instruments used in the measuring process. This means, in philosophical terms, that it is just wrong to assume, or obtain, an arbitrary accuracy when measuring the {\em real world}. There exist limitations in our knowledge power which go beyond our ability to construct more and more refined tools for measuring physical quantities. This is something intrinsic with nature, and there is no way out.
	
	Hundreds of papers have been written on this subject from very many points of view. Mathematical derivations and possible extensions and various issues can be found here, \cite{hall,mess}, while in \cite{agarwal,pita,gibi} several microscopic applications have been discussed. Recently, the role of uncertainty in other macroscopic systems, and in particular in Decision Making, has also been considered, \cite{zhang,bag2018,bag2019}. In all these cases, the observables have been assumed to be Hermitian\footnote{In this paper Hermitian and slf-adjoint will be considered as synonimous.}.
	
	The case we want to address here is what happens when $A$ and $B$ are not Hermitian. This is interesting in view of the spreading interest in non-Hermitian Quantum Mechanics, starting {\em officially} in 1994 with \cite{Bender1998}, and then considered by several authors under different aspects, both from mathematicians and from physicists. Several authors have already considered this problem in recent years, \cite{pati}-\cite{bhabani}. In this paper we will continue this analysis from our own point of view: we will first deduce different forms of the HUR for non Hermitian operators and analyze the relation between them. Then we will restrict to the case of Hermitian operators, and we show with an explicit example how the standard HUR should be corrected to cover various situations, while in Section \ref{subsectII2} we will discuss the problem of saturating the HUR. Some comments of possible extensions of the HUR to more than two operators are given in Section \ref{subsectII3}. This will produce a bound on the quantities involved in the analysis. In Section \ref{sectIII} we describe the link between the HUR and what we have recently called {\em the $\gamma$-dynamics}, relevant when the dynamics of our physical system is driven by a non Hermitian Hamiltonian, \cite{bag2022}. The role of a different scalar product is discussed in Section \ref{subsectIII1}, while in Section \ref{subsectIII2} we consider a particular choice of scalar product in presence of pseudo-bosons, and we comment on the various definitions of the Heisenberg-like dynamics in presence of non self-adjoint Hamiltonians. Section \ref{sectconcl} contains our conclusions.

	\section{First results}\label{sect2}
	
	Let $\Hil$ be an Hilbert space with scalar product $\langle.,.\rangle$ and related norm $\|.\|=\sqrt{\langle.,.\rangle}$, and let $A$ and $B$ be two non Hermitian operators which, for simplicity, we will assume for the moment to be bounded: $A,B\in\B(\Hil)$, $A\neq A^\dagger$ and $B\neq B^\dagger$. Here $\B(\Hil)$ is the set of all bounded linear operators on $\Hil$. We recall that the adjoint of $X$, $X^\dagger$, is defined by $\langle X^\dagger f,g\rangle=\langle f,Xg\rangle$, $f,g\in\Hil$, for all $X\in\B(\Hil)$. All along the paper we will often comment on what happens if some of the operators involved in our analysis is not an element of $\B(\Hil)$, since this case is relevant in some concrete physical situation. We start, as usual, by defining, taken $X\in \B(\Hil)$ and $\varphi\in\Hil$, normalized, the following quantity:
	\be
	(\Delta X)^2=\|(X-\langle X\rangle)\varphi\|^2=\langle X^\dagger X\rangle-\langle X^\dagger\rangle\langle X\rangle,
	\label{21}\en
	where $\langle X\rangle=\langle\varphi,X\varphi\rangle$.  It is clear that $\Delta X$, $\langle X\rangle$, and the other mean values, all depend on $\varphi$. To simplify the notation, we will make this dependence explicit only when needed. Calling further $\hat X=X-\langle X\rangle$, which is also $\varphi$-dependent, we can write
	\be
	\Delta X=\|\hat X\varphi\|,
	\label{22}\en
	which is, as it was understood already at an early stage of the scientific interest for HUR,  an interesting way to look at the uncertainty of $X$. In fact, we can now simply use
	the Schwarz inequality as follows: $$\Delta A\, \Delta B=\|\hat A\varphi\|\|\hat B\varphi\|\geq\left|\langle\hat A\varphi,\hat B\varphi\rangle\right|=\left|\langle \hat A^\dagger\hat B\rangle\right|,$$
	which is a first version of the uncertainty relation for $(A,B,\varphi)$. With easy computations, we can rewrite the right-hand side of this formula in two alternative ways, since
		$$
	\langle \hat A^\dagger\hat B\rangle=\langle A^\dagger B\rangle-\langle A^\dagger\rangle\langle B\rangle=\frac{1}{2}\left(\langle[A^\dagger,B]\rangle+\langle\{\hat A^\dagger,\hat B\}\rangle\right),
	$$
	where $[X,Y]=XY-YX$ and $\{X,Y\}=XY+YX$ are, respectively, the commutator and the anticommutator between $X$ and $Y$, $X,Y\in\B(\Hil)$, \cite{schr}.
	Hence we have
	\be
	\Delta A\, \Delta B\geq\left|\langle \hat A^\dagger\hat B\rangle\right|=\left| \langle A^\dagger B\rangle-\langle A^\dagger\rangle\langle B\rangle\right|=\frac{1}{2}\left|\langle[A^\dagger,B]\rangle+\langle\{\hat A^\dagger,\hat B\}\rangle\right|.
	\label{23}\en
	
	Rather than (\ref{23}), we can find another  inequality which, as we will show later, is not equivalent to the one above. The procedure, which somehow extends the one often adopted in the existing literature, is the following:
	
	\vspace{2mm}
	
	we start defining a new operator $X_\alpha=\hat A-i\alpha \hat B$, $\alpha\in\mathbb{R}$. For all normalized $\varphi\in\Hil$, (or for all $\varphi\in D(X_\alpha)$, the domain of $X_\alpha$, if this operator is unbounded, see the Remark below), we have
	\be
	\|X_\alpha\varphi\|^2=\alpha^2 \|\hat B\varphi\|^2+\alpha C_{A,B;\varphi}+\|\hat A\varphi\|^2\geq0,
	\label{24}\en
	which must be true independently of the value of $\alpha$. Here
	\be
	C_{A,B;\varphi}=i\, \langle \hat B^\dagger \hat A- \hat A^\dagger \hat B\rangle=-2\Im\langle \hat B^\dagger \hat A\rangle,
	\label{25}\en
	which is real, clearly.
		The inequality in (\ref{24}) is always satisfied if the discriminant is non positive, i.e. if $4\|\hat A\varphi\|^2\|\hat B\varphi\|^2\geq C_{A,B;\varphi}^{\,2}$. Then we have
		\be
		\Delta A\, \Delta B\geq \frac{1}{2}|C_{A,B;\varphi}|=\left|\Im\langle \hat B^\dagger \hat A\rangle\right|.
	\label{26}\en

	\vspace{2mm}

	{\bf Remark:--} If $X_\alpha$ is unbounded, then the vector $\varphi$ cannot be any vector in $\Hil$. Indeed, in this case, either $A$ or $B$, or both, are unbounded. Hence we have to consider the domains of definition of these operators. In particular, we have to check if these operators are densely defined. We will always assume in this paper that this condition is satisfied. In particular, all along this paper we will work under one of the following hypotheses: the domains are dense in $\Hil$, or they coincide with  $\Hil$ itself. Because of our particular interest, this must also be true for products of operators. In particular, if $A$ and/or $B$ are not bounded, the commutator and the anticommutator should be defined properly. For instance, if we have a common domain of $AB$ and $BA$, $\D$, dense in $\Hil$, then taken any $f\in\D$ we can put $[A,B]f=ABf-BAf$. Alternatively, one could try to work with some algebras of unbounded operators $\Lc^\dagger(\D)$ which should replace the C$^*$-algebra $\B(\Hil)$ in this particular situation, \cite{aitbook,schu,bagrev}.

	\vspace{2mm}

	Let us now introduce a second linear combination of $\hat A$ and $\hat B$, $Y_\alpha=\hat A+\alpha \hat B$, again with $\alpha\in\mathbb{R}$. In this case (\ref{24}) is replaced by a different inequality, which must again be satisfied for all values of $\alpha$:
		\be
	\|Y_\alpha\varphi\|^2=\alpha^2 \|\hat B\varphi\|^2+\alpha D_{A,B;\varphi}+\|\hat A\varphi\|^2\geq0,
	\label{27}\en
	where 	\be
	D_{A,B;\varphi}=i\, \langle \hat B^\dagger \hat A+ \hat A^\dagger \hat B\rangle=2\Re\langle \hat B^\dagger \hat A\rangle,
	\label{28}\en
	which is also real. The positivity of $\|Y_\alpha\varphi\|^2$ in (\ref{27}) is always guaranteed if 
	\be
	\Delta A\, \Delta B\geq \frac{1}{2}|D_{A,B;\varphi}|=\left|\Re\langle \hat B^\dagger \hat A\rangle\right|.
	\label{29}\en
	Putting together (\ref{26}) and (\ref{29}) we have
	\be
	\Delta A\, \Delta B\geq \frac{1}{2}\max\left\{|C_{A,B;\varphi}|,|D_{A,B;\varphi}|\right\}=\max\left\{\left|\Re\langle \hat B^\dagger \hat A\rangle\right|,\left|\Im\langle \hat B^\dagger \hat A\rangle\right|\right\},
	\label{210}\en
	which should now be compared with (\ref{23}). First of all, it is easy to prove the following result:
	
	\begin{lemma}\label{lemma1}
		The inequality in (\ref{23}) implies that in (\ref{210}). Viceversa, if $\Im\langle \hat B^\dagger \hat A\rangle\cdot\Re\langle \hat B^\dagger \hat A\rangle=0$, the inequality in (\ref{210}) implies the one in (\ref{23}).
	\end{lemma}

\begin{proof}
		The (simple) proof is based on the obvious identity $\left|\langle \hat A^\dagger\hat B\rangle\right|=\sqrt{(\Re\langle \hat B^\dagger \hat A\rangle)^2+(\Im\langle \hat B^\dagger \hat A\rangle)^2}$. Hence
	$$
	\left|\langle \hat A^\dagger\hat B\rangle\right|\geq \max\left\{\left|\Re\langle \hat B^\dagger \hat A\rangle\right|,\left|\Im\langle \hat B^\dagger \hat A\rangle\right|\right\}.
	$$
	Then (\ref{23}) clearly implies (\ref{210}).
	
	On the other hand, if we have, for instance, $\Im\langle \hat B^\dagger \hat A\rangle=0$,  then  $\left|\langle \hat A^\dagger\hat B\rangle\right|=\left|\Re\langle \hat B^\dagger \hat A\rangle\right|$. Then formula (\ref{210}) implies that
	$$
	\Delta A\, \Delta B\geq \left|\Re\langle \hat B^\dagger \hat A\rangle\right|=\left|\langle \hat A^\dagger\hat B\rangle\right|,
	$$
	which is (\ref{23}). Similar computations can be repeated if  $\Re\langle \hat B^\dagger \hat A\rangle=0$.

\end{proof}

This lemma implies what follows: since the right-hand side of (\ref{23}) is larger than or equal to the right-hand side of (\ref{210}), then, from the point of view of the uncertainty relation, what is really relevant is clearly the inequality (\ref{23}). However, if $\Im\langle \hat B^\dagger \hat A\rangle\cdot\Re\langle \hat B^\dagger \hat A\rangle=0$, there is no difference between (\ref{23}) and (\ref{210}). But, if we are interested in conditions which saturate the inequalities, then (\ref{210}) looks  more relevant than (\ref{23}): if a vector saturates (\ref{210}), then it saturates (\ref{23}) a fortiori, while the opposite is, in general, false.

\vspace{2mm}

{\bf Remark:--} It is possibly slightly more elegant to rewrite the relevant inequalities above in an equivalent form, by making use of the  operator $P_\varphi=|\varphi\rangle\langle\varphi|$, in bra-ket notation, or, more explicitly, $P_\varphi f=\langle\varphi,f\rangle\varphi$, $\forall f\in\Hil$. If we then introduce $Q_\varphi=\1-P_\varphi$, it is easy to check that both $P_\varphi$ and $Q_\varphi$ are orthogonal projectors: $P_\varphi=P_\varphi^\dagger=P_\varphi^2$, and  $Q_\varphi=Q_\varphi^\dagger=Q_\varphi^2$. It is also easy to see that, for instance, $\hat A\varphi=Q_\varphi A\varphi$. Hence formula (\ref{23}) can be rewritten as
\be
\Delta A\, \Delta B\geq\left|\langle A^\dagger Q_\varphi B\rangle\right|,
\label{23bis}\en
while (\ref{210}) can be rewritten as
\be
\Delta A\, \Delta B\geq \max\left\{\left|\Re\langle  B^\dagger Q_\varphi A\rangle\right|,\left|\Im\langle B^\dagger Q_\varphi A\rangle\right|\right\}=\max\left\{\left|\Re\langle  A^\dagger Q_\varphi B\rangle\right|,\left|\Im\langle A^\dagger Q_\varphi B\rangle\right|\right\}.
\label{210bis}\en

\vspace{2mm}

It is interesting now to study more in details these inequalities, and investigate the relations between them, at least in special cases.	
	
\subsection{Hermitian operators: $A=A^\dagger$ and $B=B^\dagger$}

This is the most common case considered in the physical literature, and not only. Under this assumption on $A$ and $B$, $A=A^\dagger$ and $B=B^\dagger$, the operator $C=-i[A,B]$ is also Hermitian. In this case $C_{A,B;\varphi}$ and $D_{A,B;\varphi}$ can be rewritten as follows:
\be
C_{A,B;\varphi}=\langle C\rangle, \qquad D_{A,B;\varphi}=\langle \{A,B\}\rangle-2\langle A\rangle\langle B\rangle,
\label{211}\en
so that formula (\ref{210}) becomes
	\be
\Delta A\, \Delta B\geq \frac{1}{2}\max\left\{|\langle C\rangle|,\left|\langle \{A,B\}\rangle-2\langle A\rangle\langle B\rangle\right|\right\}.
\label{212}\en
It is clear that this is not the inequality usually met in the literature, which only refers to $|\langle C\rangle|$. It is also clear that (\ref{212}) reduces to $\Delta A\, \Delta B\geq \frac{1}{2}|\langle C\rangle|$ when $|\langle C\rangle|\geq \left|\langle \{A,B\}\rangle-2\langle A\rangle\langle B\rangle\right|$, but not in general. In other words, it is not difficult to find examples for which (\ref{212}) is satisfied, while
\be
|\langle C\rangle|=|\langle [A,B]\rangle|\geq \left|\langle \{A,B\}\rangle-2\langle A\rangle\langle B\rangle\right|
\label{213}\en
is not.
	
{\bf Example nr. 1--} A first counterexample can be easily constructed choosing $A=B=x_0$, the (self-adjoint) multiplication operator in $\ltwo$, and $\varphi=\Phi(z)$, the standard coherent state, normalized eigenstate of the bosonic annihilation operator $c$, $c\Phi(z)=z\Phi(z)$, \cite{klauder}-\cite{bagspringer}. Recalling that $x_0=\frac{c+c^\dagger}{\sqrt{2}}$ and $[c,c^\dagger]=\1$, it is an easy exercise to check that, while $\langle [A,B]\rangle=\langle \Phi(z), [x_0,x_0]\Phi(z)\rangle=0$, 
$$
\langle \{A,B\}\rangle-2\langle A\rangle\langle B\rangle=2\langle\Phi(z), x_0^2 \Phi(z)\}\rangle-2\langle\Phi(z), x_0\,\Phi(z)\rangle^2=1,
$$
	so that (\ref{213}) does not hold. However, since another easy computation shows that $\Delta x_0=\frac{1}{\sqrt{2}} $, the inequality (\ref{212}) holds as a strict equality: coherent states saturate this inequality. Incidentally we observe that, despite of the fact that $ x_0$ is unbounded, all the quantities here are well defined, since $\Phi(z)$ belongs to the domain  of $x_0$ and of all its powers: $\Phi(z)\in\D^\infty(x_0)$, for all complex $z$. Moreover, since the set of all coherent states is total in $\ltwo$ (i.e. the only $f(x)\in\ltwo$ which is orthogonal to all the $\Phi(z)$, $z\in\mathbb{C}$, is zero), the density assumption of our Remark after (\ref{26}) is satisfied.
	
We observe that, if we rather take $A=x_0$, $B=p_0$, the self-adjoint momentum operator, and again $\varphi=\Phi(z)$, $\langle \{A,B\}\rangle-2\langle A\rangle\langle B\rangle=0$, so that (\ref{212}) returns the very well known equality $
	\Delta x_0\, \Delta p_0= \frac{1}{2}|\langle [x_0, p_0]\rangle|=\frac{1}{2}$. In other words, in this case (\ref{212}) collapses to the standard uncertainty relation for the position and the momentum operators, in its saturated version.

	\vspace{2mm}
	
	In this counterexample it might appear a little bit disturbing the fact that we are using the same operator for $A$ and $B$ since, with this choice, the left-hand side of (\ref{213}) is always automatically zero. So one may wonder what happens when $[A,B]\neq0$. Of course, in our deductions above, there is no reason to impose any constraint of this kind on $A$ and $B$. Moreover, as the next example shows, there exist situations in which $[A,B]\neq0$ and again (\ref{212}) is satisfied, while (\ref{213}) is not.

	{\bf Example nr. 2--} Let now $A=c+c^\dagger=A^\dagger$, $B=c^\dagger\,c=B^\dagger$, and $\varphi=\Phi(z)$, as in Example nr. 1. In this case, since $[A,B]=c-c^\dagger\neq0$, we get $\langle [A,B]\rangle=\langle\Phi(z), (c-c^\dagger)\Phi(z)\rangle=z-\overline z=2\,i\,y$, where $z=x+iy$. Moreover, after a sort of normal ordering procedure, consequence of the commutation rule $[c,c^\dagger]=\1$, we find
	$$
	\langle \{A,B\}\rangle=\langle\Phi(z), ( c+c^\dagger+2c^\dagger\,c^2+2{c^\dagger}^2c)\Phi(z)\rangle=(1+2|z|^2)(z+\overline{z}),
	$$
	so that $\langle \{A,B\}\rangle-2\langle A\rangle\langle B\rangle=z+\overline{z}=2x$. Hence we conclude that the left-hand side of (\ref{213}) is proportional to the imaginary part of $z$ in $\Phi(z)$, $|\langle [A,B]\rangle|=2|y|$, while its right -hand side is proportional to the real part of $z$, $\left|\langle \{A,B\}\rangle-2\langle A\rangle\langle B\rangle\right|=2|x|$. This means that everything is possible:  $|\langle [A,B]\rangle|$ can be larger, smaller or equal to $\left|\langle \{A,B\}\rangle-2\langle A\rangle\langle B\rangle\right|$, depending on the $z$ we are considering in $\Phi(z)$. 
	
	Once again, however, (\ref{212}) is satisfied, independently of our choice of $z$. Indeed we have, using the eigenvalue equation $c\Phi(z)=z\Phi(z)$ and the commutation rule $[c,c^\dagger]=\1$, $\Delta A=1$ and $\Delta B=|z|$, so that the inequality (\ref{212}) becomes
	$$
	|z|\geq \max\{|x|,|y|\},
	$$
	which is always true, for all $z\in\mathbb{C}$, and is saturated if $x\cdot y=0$, i.e. if the real or the imaginary part of $z$ (or both) are zero. However, if both $x$ and $y$ are different from zero, the one above is a strict inequality. Hence this is an example of a pair of operators $A$ and $B$ whose uncertainty relation is not saturated by coherent states, in general, differently from what happens to $\hat x$ and $\hat p$. Also, this example shows that the standard version of the HUR, $\Delta A\, \Delta B\geq \frac{1}{2}|\langle [A,B]\rangle|$, is not optimal. This is what happens here if we use as $\varphi$ a coherent state with $y=0$. In this case, $\langle [A,B]\rangle=0$ so that the inequality $\Delta A\, \Delta B\geq \frac{1}{2}|\langle [A,B]\rangle|$ returns simply $\Delta A\, \Delta B\geq 0$, which is always trivially true. If we rather consider the other term, $\left|\langle \{A,B\}\rangle-2\langle A\rangle\langle B\rangle\right|$, in the computation, we deduce a non trivial bound, $\Delta A\, \Delta B\geq |x|$. Then the role of this term in (\ref{212}) appears quite relevant, since it can be useful in getting better lower bounds for the product of the variances of $A$ and $B$.

	\subsection{Saturating (\ref{23}) and (\ref{210})}\label{subsectII2}
	
	First we consider the inequality in (\ref{23}), and we look for conditions which returns a strict equality.  We can prove the following result, which is a minor extension of an analogous proposition given in \cite{hall} for self-adjoint operators.
	
	\begin{prop}\label{prop2}
		Given two operators $A, B\in\B(\Hil)$, not necessarily self-adjoint, and a normalized vector $\varphi\in\Hil$, we have
			\be
		\Delta A\, \Delta B=\left|\langle \hat A^\dagger\hat B\rangle\right|=\left| \langle A^\dagger B\rangle-\langle A^\dagger\rangle\langle B\rangle\right|=\frac{1}{2}\left|\langle[A^\dagger,B]\rangle+\langle\{\hat A^\dagger,\hat B\}\rangle\right|,
		\label{214}\en
		if one of the following condition is satisfied:
			(c1) $\varphi$ is an eigenstate of $A$; 	(c2) $\varphi$ is an eigenstate of $B$; 	(c3) $\varphi$ is an eigenstate of $S_\gamma=A+\gamma B$, for some $\gamma\in\mathbb{C}$.
		
		Viceversa, if one among (c1), (c2) or (c3) is satisfied, then (\ref{214}) is true.
		
	\end{prop}

The proof is not particularly different from that in \cite{hall}, and will not be given here. In this case we will write that $(A,B;\varphi)$ saturates (\ref{23}). We just want to comment on two particular aspects of this proposition:

first of all, it is easy to adapt the statement to the case of unbounded $A$ and $B$. One possibility is to assume that $\varphi$ belongs to the domain of all the relevant operators. In particular, we could restrict ourselves to those $\varphi\in D(A^\dagger)\cap D(B)\cap D(A^\dagger B)\cap D(BA^\dagger)$, which, of course, should be {\em large enough}. In fact, it is sufficient to require that  $\varphi\in D(A^\dagger)\cap D(B)$, replacing $\langle A^\dagger B\rangle$ with $\langle A\varphi, B\varphi\rangle$, and so on. Alternatively, as already observed before, we could assume that $\varphi\in\D$ and that $A,B\in\Lc^\dagger(\D)$. Other possibilities can also be considered, as for instance some stability condition of a suitable dense subset of $\Hil$, see \cite{bagspringer}.

The second comment on Proposition \ref{prop2} is that, with respect to what proved in \cite{hall}, here $\gamma$ in (c3) may have both a non zero real and imaginary parts, while in \cite{hall} it  was proved to be purely imaginary. This is connected to the fact that we are no longer assuming that $A$ and $B$ are self-adjoint.
	
In Lemma \ref{lemma1} we have seen that (\ref{23}) and (\ref{210}) are equivalent only when either $\Im\langle \hat B^\dagger \hat A\rangle$ or $\Re\langle \hat B^\dagger \hat A\rangle$, or both, are zero. Otherwise, they are different, and we only know that (\ref{23}) implies (\ref{210}), while the opposite implication is false. For this reason, it is interesting to check now what can be deduced if, rather than (\ref{214}), we require that
$$
\Delta A\, \Delta B=\left|\langle \hat A^\dagger\hat B\rangle\right|=\left| \langle A^\dagger B\rangle-\langle A^\dagger\rangle\langle B\rangle\right|=\frac{1}{2}\left|\langle[A^\dagger,B]\rangle+\langle\{\hat A^\dagger,\hat B\}\rangle\right|=
$$
\be
= \frac{1}{2}\max\left\{|C_{A,B;\varphi}|,|D_{A,B;\varphi}|\right\}=\max\left\{\left|\Re\langle \hat B^\dagger \hat A\rangle\right|,\left|\Im\langle \hat B^\dagger \hat A\rangle\right|\right\},
\label{215}\en
i.e., when  $\Im\langle \hat B^\dagger \hat A\rangle\cdot\Re\langle \hat B^\dagger \hat A\rangle=0$.

We start considering the simplest situation:
\be
\langle\hat A\varphi,\hat B\varphi\rangle=0.
\label{216}\en
	In this case, $\Im\langle \hat B^\dagger \hat A\rangle=\Re\langle \hat B^\dagger \hat A\rangle=0$, so that we are in the (strongest) conditions which imply equation (\ref{215}), according to Lemma \ref{lemma1}. Now, since $\Delta A\, \Delta B=\|\hat A\varphi\|\|\hat B\varphi\|=\left|\langle \hat A^\dagger\hat B\rangle\right|=0$, then either $\hat A\varphi=0$, or $\hat B\varphi=0$, or both. This means that $\varphi$ is an eigenstate of $A$, with eigenvalue $\langle A\rangle$, or that $\varphi$ is an eigenstate of $B$, with eigenvalue $\langle B\rangle$, or both. These results are clearly in agreement with Proposition \ref{prop2}. 
	
	Let us now rather assume that 
	
	\be
	\langle \hat A\varphi,\hat B\varphi\rangle\neq0.
	\label{217}\en
	This clearly implies that both $\hat A\varphi\neq0$ and $\hat B\varphi\neq0$. If
	\be
\Re \langle \hat A\varphi,\hat B\varphi\rangle=D_{A,B;\varphi}=0,
	\label{218}\en
	then $\left|\langle \hat A^\dagger\hat B\rangle\right|=\left|\Im\langle \hat A^\dagger\hat B\rangle\right|$, which cannot be zero, of course. The fact that $\hat A\varphi\neq0$ and $\hat B\varphi\neq0$ implies that we are in condition (c3) of Proposition \ref{prop2}, so that $S_\gamma\varphi=s_\gamma\varphi$, with $s_\gamma=\langle  A\rangle+\gamma \langle B\rangle$, so that
$
\hat A\varphi=-\gamma\,\hat B\varphi$. Inserting this equation in (\ref{218}) we find that $\Re (-\gamma \|\hat B\varphi\|^2)=0$, which implies that $\gamma=i\gamma_i$, $\gamma_i\in\mathbb{R}\setminus\{0\}$, necessarily. Hence we have $
\hat A\varphi=-i\gamma_i\,\hat B\varphi$.

If we rather have
\be
\Im \langle \hat A\varphi,\hat B\varphi\rangle=D_{A,B;\varphi}=0,
\label{219}\en
then $\left|\langle \hat A^\dagger\hat B\rangle\right|=\left|\Re\langle \hat A^\dagger\hat B\rangle\right|$, which cannot be zero. In this case, repeating the same computations, we find that $
\hat A\varphi=-\gamma_r\,\hat B\varphi$, where $\gamma_r\in\mathbb{R}\setminus\{0\}$.

\vspace{2mm}
	
No particular difference arises when $A=A^\dagger$ and $B=B^\dagger$, in particular when (\ref{216}) holds. If we rather have (\ref{217}), we can observe that, if in particular (\ref{218}) is satisfied, then  $\left|\langle \hat A\hat B\rangle\right|=\left|\Im\langle \hat A\hat B\rangle\right|=\frac{1}{2}\left|\langle[A,B]\rangle\right|$, so that $\Delta A\, \Delta B=\frac{1}{2}\left|\langle[A,B]\rangle\right|$, which is the well known result for the HUR. 

If we have, again, (\ref{217}), but if (\ref{219}) is satisfied rather than (\ref{218}), then we get
$$
\Delta A\, \Delta B=\frac{1}{2}\left|\langle\{A,B\}\rangle-2\langle A\rangle\,\langle B\rangle\right|,
$$
	which has to replace the previous one in this case. Once again, we see some differences with the standard results.

	\subsection{Three operators}\label{subsectII3}
	
	It might be interesting to briefly comment on what happens, and which kind of bound on the uncertainties have to be expected, when more than two operators are considered. In other words, which is the counterpart of formula (\ref{23}) when three operators $A$, $B$ and $C$, possibly non self-adjoint, are involved? We refer to \cite{spiros} for some result in this direction, but only for self-adjoint observables which are canonically conjugate, and to \cite{qin}, again for self-adjoint observables.
	
	 As shown in (\ref{22}), the variance of each operator $X$ on $\varphi$ can be expressed as the norm of the related vector $\hat X\varphi$. For this reason, the role of the Schwarz inequality is relevant in deducing (\ref{23}), and to discuss when saturation occurs. But the Schwarz inequality only involves two vectors. To extend the inequality to, say, three vectors, $f_1$, $f_2$ and $f_3$, not necessarily normalized, one possibility is to observe that the matrix
	$$
	F_3=\left(
	\begin{array}{ccc}
		\langle f_1,f_1\rangle & \langle f_1,f_2\rangle & \langle f_1,f_3\rangle \\
		\langle f_2,f_1\rangle & \langle f_2,f_2\rangle & \langle f_2,f_3\rangle \\
				\langle f_3,f_1\rangle & \langle f_3,f_2\rangle & \langle f_3,f_3\rangle \\
	\end{array}
	\right)
	$$
	is positive: $c^\dagger F_3 c\geq0$ for all $c^T=(c_1,c_2,c_3)\in\mathbb{C}^3$. Here $c^T$ is the transpose of the column vector $c$. Hence, using the Sylvester criterion, we must have
	$$
	\langle f_1,f_1\rangle\geq0, \qquad 	\langle f_1,f_1\rangle	\langle f_2,f_2\rangle-|	\langle f_1,f_2\rangle|^2\geq0,
	$$
	and
	$$
	\|f_1\|^2\|f_2\|^2\|f_3\|^3+2\Re\left(\langle f_1,f_2\rangle\langle f_2,f_3\rangle\langle f_3,f_1\rangle\right)\geq $$
$$ \geq \|f_1\|^2|\langle f_2,f_3\rangle|^2+\|f_2\|^2|\langle f_1,f_3\rangle|^2+\|f_3\|^2|\langle f_1,f_2\rangle|^2.
$$
	The first inequality is trivially true. The second is exactly the  Schwarz inequality. The third gives a condition among the norms and the scalar product of the three vectors. Hence this is the one which extends (\ref{23}) to the case of three operators. In doing so, we get the following inequality:
	$$
	(\Delta A)^2(\Delta B)^2(\Delta C)^2+2\Re\left(\langle \hat A^\dagger \hat B\rangle\langle \hat B^\dagger \hat C\rangle\langle \hat C^\dagger \hat A\rangle\right)\geq
	$$
	\be
		(\Delta A)^2\langle \hat B^\dagger \hat C\rangle+(\Delta B)^2\langle \hat A^\dagger \hat C\rangle+(\Delta C)^2\langle \hat A^\dagger \hat B\rangle,
	\label{220}\en
		which
 produces a bound on the quantities involved, but is not as easy to use as the inequality in (\ref{23}), of course. We also refer to \cite{qin} for a very similar result for self-adjoint operators.  
 
 \vspace{2mm}
 
 {\bf Remark:--} It might be interesting to observe that it is not hard to find a different inequality involving the variances of $A$, $B$ and $C$, which is just a simple consequence of the Schwarz inequality, used three times: 
$$
|\langle f_1,f_1\rangle\langle f_2,f_3\rangle\langle f_3,f_1\rangle|\leq \|f_1\|^2 \|f_2\|^2 \|f_3\|^2,
$$ 
	which in turns produces
	\be
	(\Delta A)^2(\Delta B)^2(\Delta C)^2\geq \left|\langle \hat A^\dagger \hat B\rangle\langle \hat B^\dagger \hat C\rangle\langle \hat C^\dagger \hat A\rangle\right|.
	\label{221}\en
	We recall that, using the equalities in (\ref{23}), this formula can be written in different ways. In particular, if $A$, $B$ and $C$ are self-adjoint and if, say, $[A,B]=0$, then (\ref{221}) returns the obvious inequality $\Delta A\Delta B\Delta C\geq0$. However we can still find some relevant bound on the uncertainties: if we also have, for instance, $[A,C]\neq0$, then $\Delta A\Delta C\geq \frac{1}{2}|\langle[A,C]\rangle|$.

	\section{Uncertainty relation and $\gamma$ dynamics}\label{sectIII}

	Once we have an uncertainty relation, and we have conditions which guarantee that this is saturated, it is also interesting to check if or when this {\em saturation} is preserved during time evolution. For instance, standard coherent states are stable under time evolution for the quantum harmonic oscillator, and it is well known that they saturate the HUR for $\hat x$ and $\hat p$, which are both self-adjoint. Hence this saturation is preserved during time evolution. However, in the situation which is more interesting for us, the operators we considered are not necessarily required to be self-adjoint. In particular, we are interested in discussing what happens if the Hamiltonian $H$ of the physical system $\Sc$ is not self-adjoint, i.e., if $H\neq H^\dagger$. This is an interesting situation, mainly in connection with $PT$- or pseudo-hermitian quantum mechanics, see \cite{mosta}-\cite{bender} and references therein.
	
	At a first view, the problem of the time evolution for these systems does not look particularly different from the standard case, at least when using the Schr\"odinger evolution: the equation to be solved is the usual one, $i\dot\psi(t)=H\psi(t)$. However, here, $H\neq H^\dagger$. This implies that, even in the simplest case in which all the operators involved in the description of $\Sc$ are bounded, the Heisenberg dynamics $\gamma^t(X)=e^{iH^\dagger t}Xe^{-iHt}$  is no longer an automorphism of $\B(\Hil)$: $\gamma^t(XY)\neq\gamma^t(X)\gamma^t(Y)$, for generic $X,Y\in\B(\Hil)$. We refer to \cite{bag2022} for several introductory results on what we have called $\gamma$-{\em dynamics}. Here we list only few facts which are relevant for us.
	
	\begin{enumerate}
		\item As already seen, for reasons which are clarified in \cite{bag2022} and in references therein, we introduce $\gamma^t(X)=e^{iH^\dagger t}Xe^{-iHt}$. Here we only observe that this is the analogous of the Heisenberg dynamics $\alpha^t(X)=e^{iH_0t}Xe^{-iH_0t}$, considered when $H_0=H_0^\dagger$. We stress once more that, here and in the following, to simplify the treatment  we will assume that all the relevant operators considered ($X$, $H$, $H_0$,...) are bounded.
		\item  We define a map $\delta_\gamma:\B(\Hil)\rightarrow\B(\Hil)$ as follows:
		\be
		\delta_\gamma(X)=\|.\|-\lim_{t,0}\frac{\gamma^t(X)-X}{t}=i\left(H^\dagger X-XH\right),
		\label{32}\en
		$X\in\B(\Hil)$. This is called a $\gamma$-derivation.
		\item The series $\sum_{k=0}^{\infty}\frac{t^k\delta_\gamma^k(X)}{k!}$ is norm convergent to $\gamma^t(X)$, for all $X\in\B(\Hil)$. Here we have used the following notation:  $\delta_\gamma^0(X)=X$, and $\delta_\gamma^k(X)=\delta_\gamma(\delta_\gamma^{k-1}(X))$, $k\geq1$.
		\item The following statements are equivalent: 1) $\delta_\gamma$ is a *-derivation, \cite{br1,br2}; 2) $\delta_\gamma(\1)=0$; 3) $H=H^\dagger$; 4) $\gamma^t(\1)=\1$; 5)  $\gamma^t(XY)=\gamma^t(X)\gamma^t(Y)$, $\forall X,Y\in\B(\Hil)$.
		
		This result implies that $\gamma^t$ cannot be an authomorphism of $\B(\Hil)$ if any of the above properties (and therefore all) is violated. We refer to \cite{bag2022} for the (serious) consequences of this lack of the authomorphism property for $\gamma^t$ when discussing the dynamics of $\Sc$.
		\item We introduce the operator $S^{-1}=\sum_{k=0}^N|\varphi_k\left>\right<\varphi_k|$, where each $\varphi_k$ is an eigenstate of $H$, with eigenvalue $E_k$. For simplicity we are  restricting here to finite sums\footnote{Or, more in general, to norm-converging series.}. In other words, we are assuming that $\dim(\Hil)=N$. If each eigenvalue of $H$ is non degenerate, then $\F_\varphi=\{\varphi_k,\,k=1,2,\ldots,N\}$
		 is a basis for $\Hil$, and admits an unique biorthonormal basis $\F_\psi=\{\psi_k,\,k=1,2,\ldots,N\}$, whose vectors are eigenstates of $H^\dagger$ with eigenvalues $\overline E_k$: $H^\dagger\psi_k=\overline{E_k}\psi_k=E_k\psi_k$, if $E_k$ is real. Then $S=\sum_{k=0}^N|\psi_k\left>\right<\psi_k|$, and $S\varphi_k=\psi_k$, $S^{-1}\psi_k=\varphi_k$, and $SH=H^\dagger S$ (when all the eigenvalues are real, which is not always the case, as it happens for $PT$-symmetric Hamiltonians in a $PT$-broken phase). This is a somewhat general settings for non self-adjoint Hamiltonians, which can be extended adding some extra technical assumptions to $\dim(\Hil)=\infty$.
		 \item We say that $X\in\B(\Hil)$ is a $\gamma$-symmetry if $[H,S^{-1}X]=0$. Notice that, if $H=H^\dagger$ and $\F_\varphi$ is an orthonormal basis, this means that $X$ commutes with $H$, since $S=S^{-1}=\1$. Notice also than $S$ is always a $\gamma$-symmetry, while $S^{-1}$, in general, is not. We have the following:  $X$ is a $\gamma$-symmetry if and only if any of the following statements, all equivalent, are satisfied: 1) $ [H^\dagger,X^\dagger S^{-1}]=0$; 2)  $H^\dagger X=XH$; 3)  $\delta_\gamma(X)=0$; 4) $\gamma^t(X)=X$. 
		 
		 This last result is particularly relevant since it implies that all $\gamma$-symmetries do not evolve in time.
		 \item If $X\in\B(\Hil)$ is $\gamma$-symmetry, and $Y\in\B(\Hil)$ is an operator commuting with $H$, then $XY$ is also a $\gamma$-symmetry.		
	\end{enumerate}
	
	We can now state the following result:
	
	\begin{prop}\label{prop3}
		If $(A,B;\varphi)$ saturates (\ref{23}), and if $A$ and $B$ are $\gamma$-symmetries, then $(\gamma^t(A),\gamma^t(B);\varphi)$ saturates (\ref{23}), $\forall t\geq0$.
	\end{prop}

In view of our comment 6. above, there is not much to prove, indeed. It is more interesting to notice that the inverse statement is not true. In other words, if  $(\gamma^t(A),\gamma^t(B);\varphi)$ saturates (\ref{23}), $\forall t\geq0$, this does not imply that $A$ and $B$ are $\gamma$-symmetries. In fact, in view of Proposition \ref{prop2}, we conclude that one of the following facts must be  satisfied for all $t$: (c1) $\varphi$ is an eigenstate of $\gamma^t(A)$, or (c2)  $\varphi$ is an eigenstate of $\gamma^t(B)$, or yet (c3)  $\varphi$ is an eigenstate of $\gamma^t(A)+\alpha\gamma^t(B)$, for some complex $\alpha$. If we now consider an operator $V_\varphi$ such that $V_\varphi\varphi=\varphi$, then it is easy to see that $\gamma^t(A)=AV_\varphi$ and $\gamma^t(B)=BV_\varphi$ saturate (\ref{23}), if $(A,B;\varphi)$ saturates (\ref{23}) at $t=0$. Hence $A$ and $B$ need not being invariant under $\gamma$-time evolution.

\vspace{2mm}

{\bf Remark:--} A class of operators like $V_\varphi$ can be easily constructed as follows: let us consider an orthonormal basis of $\Hil$, $\F_e=\{e_n, n\geq1\}$, such that $e_1=\varphi$. This is always possible. Let $c=\{c_n\}\in\l^2(\mathbb{N})$ be a complex-valued sequence, with $c_1=1$. We call $V_\varphi f=\sum_{n=1}^\infty c_n\langle e_n,f\rangle e_n$. It is clear that $V_\varphi\varphi=\varphi$, and that $\|V_\varphi\|\leq \|c\|_2$, the $l^2$-norm of the sequence $c$. Different choices of $c$ and $e_n$, $n\geq2$, give rise to different explicit forms of the operator $V_\varphi$.
	
	\vspace{2mm}
	
It would be interesting to check if the assumption that $(\gamma^t(A),\gamma^t(B);\psi)$ saturates (\ref{23}) for all possible normalized $\psi\in\Hil$ and for all $t\geq0$ implies that $A$ and $B$ are $\gamma$-symmetries. This is still not so clear, and it is work in progress.

A simple consequence of what we have discussed so far is that, if $A$ and $B$ are $\gamma$-symmetries, and if there exists $C\in\B(\Hil)$ such that $[C,H]=0$ and $(AC,BC;\varphi)$ saturate (\ref{23}), then  $(\gamma^t(AC),\gamma^t(BC);\varphi)$ saturate (\ref{23}) for all $t\geq0$. This is simply because both $AC$ and $BC$ are $\gamma$-symmetries, and because of Proposition \ref{prop3}.

\vspace{2mm}

{\bf Remark:--} In what we have discussed here, the effect of the dynamics is of the {\em Heisenberg-type}: the state of the system (i.e., here, the normalized vector of $\Hil$) is time-independent, while the observables are not. If we change representation, going to the Schr\"odinger representation, the opposite happens it is exactly the state that evolves in time. And this evolution can transform a pure state into some linear combination of states. This suggests that the role of mixed states could be analyzed, in our dynamical context. This is an interesting line of research, see for instance \cite{luo}, and we hope to be able to work on it soon.

\subsection{A different scalar product}\label{subsectIII1}

It is widely discussed in the literature that, in presence of non self-adjoint Hamiltonians, it could be useful to introduce in the Hilbert space $\Hil$ a different scalar product, which is such that many of the properties of a quantum system driven by a self-adjoint Hamiltonian are recovered. We refer to \cite{baghat2021}, and to references therein, for some details on this aspect, and for a particular construction of several scalar products useful to produce {\em exactly solvable} Hamiltonians\footnote{With this we mean Hamiltonians whose eigensystems can be found explicitly, and exactly.}.

In particular, as we have sketched before, given $H\neq H^\dagger$, we can introduce an operator $S$ satisfying $SH=H^\dagger S$. Here we are assuming, as often in this paper, that $S,H\in\B(\Hil)$. $S$ is positive and self-adjoint, it admits an unique positive, and self-adjoint, square root $S^{1/2}$, which admits bounded inverse: $S^{-1/2}\in\B(\Hil)$. Of course, this means that $S^{-1}\in\B(\Hil)$ as well. 

We can now use $S$ to define the following scalar product and its related norm
\be
\langle f,g\rangle_S=\langle S\,f,g\rangle, \qquad \|f\|_S^2=\langle f,f\rangle_S=\|S^{1/2}f\|^2,
	\label{33}\en
	where $f,g\in\Hil$. This is indeed a scalar product on $\Hil$ and we can introduce an adjoint $\sharp$ as follows:
	\be
	\langle X f,g\rangle_S=\langle f,X^\sharp g\rangle_S,
	\label{34}\en
	for $X$ bounded on $\Hil$ and $f,g\in\Hil$.

	Now we can extend (\ref{21}):
	\be
	(\Delta X)_S^2=\|(X-\langle X\rangle_S)\varphi\|_S^2=\|\tilde X\varphi\|_S^2
	\label{35}\en
	where $\langle X\rangle_S=\langle\varphi,X\varphi\rangle_S$ and $\tilde X=X-\langle X\rangle_S$. Applying  the Schwarz inequality for $\langle .,.\rangle_S$, we find that
	\be
	(\Delta A)_S(\Delta B)_S=\|\tilde A\varphi\|_S\|\tilde B\varphi\|_S\geq |\langle\tilde A\varphi,\tilde B\varphi\rangle_S|=|\langle\varphi,\tilde A^\sharp\tilde B\varphi\rangle_S|.
	\label{36}\en
	It is known, and easy to check, that if $A$ satisfies $SA=A^\dagger S$, then $A=A^\sharp$ and $\tilde A=\tilde A^\sharp$. This is what happens, in particular, when $A=H$. So the right-hand side of (\ref{36}) becomes $|\langle\varphi,\tilde A\tilde B\varphi\rangle_S|$. If also $B$ satisfies the same equality, $SB=B^\dagger S$, then we find the following alternative form of (\ref{23}) and (\ref{26}):
	\be
(	\Delta A)_S\, (\Delta B)_S\geq\left|\langle \hat A\hat B\rangle_S\right|=\frac{1}{2}\left|\langle[A,B]\rangle_S+\langle\{\tilde A^\dagger,\tilde B\}\rangle_S\right|\geq \frac{1}{2}\left|\langle[A,B]\rangle_S\right|,
	\label{37}\en
	repeating here similar computations.
	
	In deducing this inequality, we have assumed that both $A$ and $B$ satisfy a certain intertwining relation with $S$ which, as we have commented before, is a $\gamma$-symmetry. Operators of this kind have been recently called {\em good observables}, \cite{bhabani}. It is natural to ask if or when two such operators really exist. In fact, it is quite easy to describe a constructive approach. First of all, we identify $A=H$ since, in this way and with the construction summarized before, the identity $SH=H^\dagger S$ is automatic. Moreover, if we introduce the new operator $H_0=S^{1/2}HS^{-1/2}$, we can easily check that $H_0=H_0^\dagger$. This is possible if all the eigenvalues of $H$ are real, while it is not true if even a single eigenvalue has a non zero imaginary part, \cite{baghat2021}. Now, let us consider {\bf any} $B_0=B_0^\dagger\in\B(\Hil)$. We can check that $B=S^{-1/2}B_0S^{1/2}$ is a good observable, i.e. that $SB=B^\dagger S$. Moreover, since
	$
	[B,H]=S^{-1/2}[B_0,H_0]S^{1/2},
	$
	it follows that $B$ commutes with $H$ if and only if $B_0$ commutes with $H_0$. This can be interpreted as follows: any quantum mechanical system $\Sc_0$ described by a self adjoint Hamiltonian $H_0$, with a symmetry $B_0$, can be {\em deformed} maintaining some of the essential features of $\Sc_0$. The operator $S-\1$ is a measure of the {\em degree of deformation}.

	\subsection{The scalar product $\langle.,.\rangle_S$ and the non Hermitian position and momentum operators}\label{subsectIII2}
	
	As we have already commented in Example nr. 1, if we consider the two self-adjoint operators $x_0=\frac{c+c^\dagger}{\sqrt{2}}$ and $p_0=\frac{c-c^\dagger}{\sqrt{2}\,i}$ on $\ltwo$, and we take $\varphi=\Phi(z)$ in (\ref{21}), we find that $
	\Delta x_0\, \Delta p_0= \frac{1}{2}$.  
	
	In the last years, mainly in connection with pseudo-hermitian (and $PT$-symmetric) quantum mechanics, it has been understood that in some cases $c$ and $c^\dagger$ are better replaced by two different operators, $a$ and $b$, such that $b^\dagger$ is, in general, different from $a$, and that $[a,b]=\1$ is satisfied (in the sense of unbounded operators). This is because many physical systems introduced in the past 20 years can be rewritten in terms of this kind of operators, \cite{bagspringer}. $a$ and $b$, together with their adjoints, act as ladder operators and produce two biorthonormal families of vectors of $\Hil$, $\F_\varphi=\{\varphi_n, \,n\geq0\}$ and $\F_\psi=\{\psi_n, \,n\geq0\}$, which may, or may not, be bases. Still they are, in all the examples under control, complete in $\Hil$. In particular $a$ is a lowering operator for $\F_\varphi$, while $b^\dagger$ is a lowering operator for $\F_\psi$: $a\varphi_n=\sqrt{n}\,\varphi_{n-1}$, $b^\dagger\psi_n=\sqrt{n}\,\psi_{n-1}$, $n\geq0$, with the understanding that $\varphi_{-1}=\psi_{-1}=0$. Because of this, it is probably not surprising that we can introduce two {\em bi-coherent states} satisfying 
	\be
	a\varphi(z)=z\varphi(z), \qquad b^\dagger\psi(z)=z\psi(z),
	\label{38}\en
	$z\in\mathbb{C}$.  We refer to \cite{bagspringer} and references therein for many results on these {\em pseudo-bosonic operators}, and their related bi-coherent states. What is interesting for us, here, is to see what does it happen if we look for the (generalized) HUR for the non self-adjoint operators $X$ and $P$ which extend $x_0$ and $p_0$ replacing $c$ with $a$ and $c^\dagger$ with $b$:
	$$
	X=\frac{a+b}{\sqrt{2}}, \qquad P=\frac{a-b}{\sqrt{2}\,i}.
	$$
	It is clear that, in general, $X\neq X^\dagger$ and $P\neq P^\dagger$. Still, $[X,P]=i\1$ (again, in the sense of unbounded operators). It is possible to show that, while $
	\Delta X\, \Delta P\neq \frac{1}{2}$, we have
	\be
	(\Delta X)_S\, (\Delta P)_S= \frac{1}{2}|\langle[X,P]\rangle_S|=\frac{1}{2},
	\label{39}\en
	where the mean values are computed on $\varphi(z)$. We prove this equality in the simplest situation, that is for {\em regular} pseudo-bosons, \cite{bagspringer}. In this case we have a bounded and invertible operator $R$, with bounded inverse, such that $a$ and $b$ are connected to $c$ and $c^\dagger$ as follows:
	\be
	a=RcR^{-1}, \qquad b=Rc^\dagger R^{-1}.
	\label{310}\en
	
	\vspace{2mm}
	
	{\bf Remark:--} This is an useful shorthand notation which should be understood in the sense of unbounded operators, introducing, for instance, the algebra $\Lc^\dagger(\D)$ as we proposed before, \cite{aitbook,schu}. Alternatively,  see \cite{bagspringer}, we can suppose that a set $\D$, dense in $\Hil$ and stable under the action of $R,R^{-1},c,c^\dagger, a, b,....$ exists, and replace $a=RcR^{-1}$ with 	$af=RcR^{-1}f$, for all $f\in\Hil$. Many examples do indeed satisfy this apparently strong assumption, \cite{bagspringer}.
	
	\vspace{2mm}
	
	In this situation it is possible to find that $\varphi(x)$, $\psi(z)$ and $\Phi(z)$ are all connected by $R$:
	\be
	\varphi(z)=R\Phi(z), \qquad \psi(z)=(R^{-1})^\dagger \Phi(z), \quad \Rightarrow \quad \varphi(z)=RR^\dagger\psi(z).
	\label{311}\en  
	Similar equalities relate the vectors in $\F_\varphi$ with those in $\F_\psi$. For instance, $\varphi_n=RR^\dagger\psi_n$. Hence, going back to point 5. of our list of results at the beginning of Section \ref{sectIII}, we can identify $S^{-1}$ with $RR^\dagger$. Now we have
	$$
	\langle[X,P]\rangle_S=\langle S\varphi(z),[X,P]\varphi(z)\rangle=i\langle\psi(z),\varphi(z)\rangle=i,
	$$
	using for instance (\ref{311}) and the normalization of $\Phi(z)$. Also
	$$
	\langle X\rangle_S=\langle S\varphi(z),X\varphi(z)\rangle=\frac{1}{\sqrt{2}}\langle \psi(z),(a+b)\varphi(z)\rangle=\frac{z+\overline z}{\sqrt{2}},
	$$
	and
	$$
	\langle P\rangle_S=\langle S\varphi(z),P\varphi(z)\rangle=\frac{1}{\sqrt{2}\,i}\langle \psi(z),(a-b)\varphi(z)\rangle=\frac{z-\overline z}{\sqrt{2}\,i},
	$$
	which are both real.
	Using further formulas (\ref{310}) and (\ref{311}), we also find that
	$$
	\langle X\varphi(z),X\varphi(z)\rangle_S=\frac{1}{2}\left(z^2+\overline{z}^2+2|z|^2+1\right),
	$$
	and
		$$
	\langle P\varphi(z),P\varphi(z)\rangle_S=-\,\frac{1}{2}\left(z^2+\overline{z}^2-2|z|^2-1\right).
	$$
	Summarizing then
	$$
	(\Delta X)_S^2=\langle X\varphi(z),X\varphi(z)\rangle_S-\langle X\rangle_S^2=\frac{1}{2}, \qquad
	(\Delta P)_S^2=\langle P\varphi(z),P\varphi(z)\rangle_S-\langle P\rangle_S^2=\frac{1}{2},
	$$
	so that (\ref{39}) follows. Hence one might be tempted to conclude that, in situation like the one described here, the {\em natural} scalar product to consider on $\Hil$ is not $\langle.,.\rangle$, but rather $\langle.,.\rangle_S$. This is, indeed, very much debated in the literature. Indeed, while the above computation, and the validity of (\ref{39}),  goes in this direction, it is easy to see that if we replace $X$ and $P$ with $X^\dagger$ and $P^\dagger$, which also satisfy $[X^\dagger,P^\dagger]=i\1$,  to recover something like equality (\ref{39}) we have to replace $\varphi(z)$ with $\psi(z)$, and $S$ with $S^{-1}$. In other words: for these other pair of operators we have to consider a third scalar product and another bi-coherent state, if we are interested in saturating the inequality. This is unpleasant, since it suggests that, when one works with non self-adjoint operators, there could be the need of using simultaneously different scalar products. This is not really what one would like to have, when describing some concrete physical system $\Sc$: the mathematical settings needed for the analysis of $\Sc$, and in particular the Hilbert space and its scalar product, should be fixed from the very beginning, and should be not changed as a consequence of what observables we are interested in. But, again, this is also an aspect which is debated, mainly by the community of people working with non self-adjoint operators.

	There is also another result that shows how the choice of the scalar product, but also of the definition of an Heisenberg-like dynamics, is indeed not so clear. Suppose for instance that our physical system is described by the manifestly non Hermitian Hamiltonian $H=\omega ba$, $\omega\in\mathbb{R}$. This is not so rare, indeed, \cite{bagabook}. Hence, \cite{bagspringer}, $e^{-iHt}\varphi(z)=\varphi(ze^{-i\omega t})$, and $e^{iH^\dagger t}\psi(z)=\psi(ze^{i\omega t})$: bi-coherent states are mapped into bi-coherent states. However, in general,
	$$
	\langle\varphi(z),\gamma^t(A)\varphi(z)\rangle_S\neq 	\langle\varphi(ze^{-i\omega t}),A\varphi(ze^{-i\omega t})\rangle_S,
	$$
	while the two sides of this formula coincide if $R=\1$, i.e. if $H=\omega c^\dagger c=H^\dagger$, in particular. If we rather introduce $\alpha^t(A)=e^{i H t}A e^{-iH t}$, despite the fact that $e^{\pm iHt}$ is not unitary, $\alpha^t$ satisfies the following important identity $
	\alpha^t(AB)=\alpha^t(A)\alpha^t(B),
	$
	as well as
	$
	\langle\varphi(z),\alpha^t(A)\varphi(z)\rangle_S=	\langle\varphi(ze^{-i\omega t}),A\varphi(ze^{-i\omega t})\rangle_S,
	$
	for all operators $A$ and $B$. So the conclusion is that, in agreement with what already discussed in \cite{bagspringer,bagaop20151,bagaop20152} and references therein, it is still not completely clear what should be called the {\em Heisenberg-like dynamics} for non Hermitian Hamiltonians. There are pros and cons in any choice one can make, and there is no special reason, in our opinion, to choose one rather than another approach.

	\section{Conclusions and perspectives}
	
	\label{sectconcl}
	
	In this paper we have considered some aspects of the HUR  mostly from the point of view of non self-adjoint operators. Some equivalence results, like those in Lemma \ref{lemma1}, and some refinements of the inequality, have been deduced. In particular we have discussed how, also for self-adjoint operators, an extra term should be added in the HUR to cover even simple situations. Also, conditions for saturating the inequalities have been discussed, in the same spirit as in \cite{hall}. 
	
	It is maybe interesting to stress that our results might appear close to others, even recent, appeared in the literature, as those, to cite quite a recent paper, in \cite{bhabani}, but the approach and the mathematical settings are indeed rather different. Apart the topic of the papers, which is the same (HUR for non self-adjoint operators), the content of our paper is really different from that in \cite{bhabani}, where a two-dimensional (and therefore bounded) model is considered. Here, recalling that most of the relevant operators in quantum mechanics live in infinite-dimensional Hilbert spaces (and are often  unbounded), we discuss a setup which works well in these cases. Moreover, here we have considered the relation between HUR and $\gamma$-dynamics, and we have connected the notion of $\gamma$-symmetries, which was introduced in \cite{bag2022}, with the notion of {\em good observables}, \cite{bhabani}.  These aspects, together with other connections with symmetries for non self-adjoint operators, still needs to be fully understood. These are part of our future plans. It is particularly interesting to achieve a better understanding of the role of the Heisenberg dynamics in situations like ours, where the Hamiltonian is not self-adjoint. There are indeed several possibilities to adopt, one different from the other, as briefly sketched in the last part of Section \ref{sectIII}. This is also an interesting aspect to analyze in details, for its consequences in many different aspects of, e.g., pseudo-hermitian quantum mechanics, from dynamics to transition probabilities.

	\section*{Acknowledgements}
	
The author acknowledges partial financial support from Palermo University (via FFR2021 "Bagarello") and from G.N.F.M. of the INdAM. 
	
	\section*{Funding statement}
	
	This work received no financial support.

	

\begin{thebibliography}{99}
		
		
		\bibitem{heis} W. Heisenberg, {\em \"Uber den anschaulichen inhalt der quantentheoretischen kinematik und mechanik}, Zeit.
		f\"ur Physik, {\bf 43}, 172-198 (1927)
		
	\bibitem{robe}	H. P. Robertson, {\em The uncertainty principle}, Phys. Rev.
{\bf 34}, 163-164 (1929)

\bibitem{schr} E. Schr\"odinger, {\em About Heisenberg uncertainty relation}, Proceedings of The Prussian Academy of Sciences
Physics-Mathematical Section,  XIX, pp.296-303, (1930)
		
		\bibitem{hall} B. C. Hall, {\em Quantum theory for mathematicians}, Springer, New York (2013)
		
		
		\bibitem{mess} A. Messiah, {\em Quantum mechanics}, vol. 2, North Holland Publishing Company, Amsterdam, (1962)
		
		
		
		
		
		
		\bibitem{agarwal} G. S. Agarwal, {\em Heisenberg’s Uncertainty Relations and Quantum Optics}, Forts. der Phys., {\bf 50}, 575-582 (2002)
		
		\bibitem{pita} L. Pitaevskii, S. Stringari, {\em Uncertainty principle, quantum fluctuations and broken symmetries}, J. Low Temp., {\bf 85}, N. 5/6, 377-388 (1991)
		
		
	\bibitem{gibi}	P. Gibilisco, D. Imparato, T. Isola, {\em Uncertainty principle and quantum Fisher information, II}, J. Math. Phys.,
{\bf 48}, 072109 (2007)
		
		\bibitem{zhang} J.-J. Zhang, D.-X. Zhang, J.-N. Chen, L.-G. Pang, D. Meng,  {\em On The Uncertainty Principle of Neural Networks}, doi:10.48550/arXiv.2205.01493
		
		\bibitem{bag2018} F. Bagarello, I. Basieva, A. Khrennikov, E. Pothos {\em Quantum like modeling of decision making: quantifying uncertainty with the aid of Heisenberg-Robinson inequality},    Journal of Mathematical  Psychology, {\bf 84}, 49-56 (2018)
		
		\bibitem{bag2019} F. Bagarello,  {\em A  dynamical approach to compatible and incompatible questions},  Physica A, {\bf 527}, 121282 (2019)
		
		
		
		
		\bibitem{Bender1998} C. M. Bender, S. Boettcher,  {\em Real
		Spectra in Non-Hermitian Hamiltonians Having PT Symmetry},  Phys.
		Rev. Lett. \textbf{80}, 5243--5246 (1998)
		
		\bibitem{pati} A. K. Pati, U. Singh, U. Sinha, {\em Measuring non-Hermitian operators via weak values}, Phys. Rev. A, {\bf 92}, 052120 (2015)
		
			\bibitem{yanagi} K. Yanagi, K. Sekikawa, {\em Non-hermitian extensions of Heisenberg
type and Schr\"odinger type uncertainty
relations}, J. Ineq. Appl, (2015) 2015:381
		
	\bibitem{fan} Y. Fan, H. Cao, W. Wang, H. Meng, L. Chen,	{\em Non-Hermitian extensions of uncertainty relations with
generalized metric adjusted skew information}, Quantum Inf. Proc. (2019) 18:309

\bibitem{calderon} G. Garcia-Calderon, J. Villavicencio, {\em Heisenberg uncertainty relations for the non-Hermitian resonance-state solutions to the
Schr\"odinger equation}, Phys. Rev A, {\bf 99}, 022108 (2019)
		
		
	\bibitem{bhabani} N. Shukla, R. Modak, B. P. Mandal, 	{\em Uncertainty Relation for Non-Hermitian Systems}, Phys. Rev. A, {\bf 107}, 042201 (2023)
	
	
	\bibitem{bag2022} F. Bagarello, {\em Heisenberg dynamics for non self-adjoint Hamiltonians: symmetries and derivations},    MPAG, {\bf 26}, 1-15 (2023)
	
		
		
		\bibitem{aitbook}  J.-P.
	Antoine, A. Inoue and  C. Trapani {\it Partial *-algebras and Their
		Operator Realizations}, Kluwer, Dordrecht, 2002
	
	\bibitem{schu} K. Schm\"udgen, {\it Unbounded operator algebras
		and Representation theory}, Birkh{\"a}user, Basel, 1990
	
	\bibitem{bagrev} F. Bagarello {\em Algebras of unbounded operators and physical
		applications: a survey},  Reviews in Math. Phys, , {\bf 19}, No. 3,
	231-272 (2007)
		
				
		\bibitem{klauder} J. R. Klauder, B. S. Skagerstam Eds., {\em Coherent states. Applications in physics and mathematical physics}, World Scientific, Singapore (1985)
		
		\bibitem{perelomov} A. M. Perelomov, {\em Generalized coherent states and their applications}, Springer-Verlag, Berlin (1986)
		
		
		\bibitem{aagbook}  S.T. Ali, J-P.  Antoine and  J-P.  Gazeau,
		{\em  Coherent States, Wavelets and Their Generalizations\/},
		Springer-Verlag, New York, (2000)
		
		\bibitem{gazeaubook}  J-P.  Gazeau, {\em Coherent states in quantum physics}, Wiley-VCH, Berlin (2009)
		
		
		\bibitem{didier} M. Combescure,  R. Didier, {\em Coherent States and Applications in Mathematical Physics},   Springer, (2012)
		
		
		
		
		\bibitem{abgspecissue} J.-P. Antoine, F. Bagarello and J.-P. Gazeau Eds, {\em Coherent states	and applications: a contemporary panorama}, Springer Proceedings in Physics, (2018)
		
			\bibitem{bagspringer} F. Bagarello, {\em Pseudo-Bosons and Their Coherent States}, Springer, Mathematical Physics Studies, 2022
		
\bibitem{spiros} S. Kechrimparis, S. Weigert, {\em Heisenberg Uncertainty Relation for Three Canonical Observables}, Phys. Rev. A, {\bf 90}, 062118 (2014)		
		
			
		\bibitem{qin} H.-H. Qin, S.-M. Fei, X. Li-Jost,{\em Multi-observable Uncertainty Relations in Product Form of Variances}, Sci. Rep., {\bf 6}, 31192 (2016)
		
		
		
		
			
		
		
		
		

	\bibitem{mosta} A. Mostafazadeh, {\em Pseudo-hermitian quantum mechanics},  Int. J. Geom. Methods Mod. Phys., {\bf 7}, 1191-1306 (2010)



\bibitem{specissue2012} C. Bender, A. Fring, U. G\"nther, H. Jones Eds, {\em Special issue on quantum physics with non-Hermitian operators}, J. Phys. A: Math. and Ther., {\bf 45} (2012)


\bibitem{bagabook} F. Bagarello, J. P. Gazeau, F. H. Szafraniec e M. Znojil Eds., {\em Non-selfadjoint operators in quantum physics: Mathematical aspects}, John Wiley and Sons (2015)


\bibitem{bagprocpa} F. Bagarello, R. Passante, C. Trapani, {\em Non-Hermitian Hamiltonians in Quantum Physics;
	Selected Contributions from the 15th International Conference on Non-Hermitian
	Hamiltonians in Quantum Physics}, Palermo, Italy, 18-23 May 2015, Springer (2016)

\bibitem{bender} C. M. Bender, {\em PT Symmetry in quantum and classical physics}, World Scientific, (2019)

\bibitem{br1} O. Bratteli and D.W. Robinson, {\em Operator
	algebras and Quantum statistical mechanics 1}, Springer-Verlag, New
York, 1987.

\bibitem{br2} O. Bratteli and D.W. Robinson, {\em Operator
	algebras and Quantum statistical mechanics 2}, Springer-Verlag, New
York, 1987.


\bibitem{luo} S. Luo, {\em Heisenberg uncertainty relation for mixed states}, Phys. Rev.
A, {\bf 72}, 042110 (2005)


	\bibitem{baghat2021} F. Bagarello,  N. Hatano, {\em A chain of solvable non-Hermitian Hamiltonians constructed by a series of metric operators},  Ann. of Phys., {\bf 430}, 168511 (2021)


		
	\bibitem{bagaop20151} F. Bagarello, {\em Some results on the dynamics and  transition probabilities for non self-adjoint hamiltonians},  Ann. of Phys., {\bf 356}, 171-184 (2015)
	
	
	
	
	\bibitem{bagaop20152} F. Bagarello, {\em Transition probabilities for non self-adjoint Hamiltonians in infinite dimensional Hilbert spaces},
	Ann. of Phys., {\bf 362}, 424-435 (2015)
	
		
	\end{thebibliography}
\end{document}